\documentclass[10pt,twocolumn,twoside,english]{IEEEtran}
\usepackage[T1]{fontenc}
\usepackage{mathrsfs}
\usepackage{bm}
\usepackage{amsmath}
\usepackage{breqn}
\usepackage{amsthm}
\usepackage{amssymb}
\usepackage{graphicx}
\PassOptionsToPackage{normalem}{ulem}
 \usepackage{subfig}
\usepackage{cite}
\makeatletter

\theoremstyle{plain}
\newtheorem{theorem}{\protect\theoremname}
  \theoremstyle{plain}
  \newtheorem{conjecture}{\protect\conjecturename}
  \theoremstyle{plain}
  
  \theoremstyle{plain}
   \newtheorem{lemma}{\protect\lemmaname}
  \theoremstyle{remark}

\theoremstyle{assumption}
    \theoremstyle{proposition}
  \newtheorem{proposition}{\protect\propositionname}
  
\theoremstyle{algorithm}

\setkeys{Gin}{width=1.0\columnwidth}

\makeatother

\usepackage{babel}
  \providecommand{\definitionname}{Definition}
  \providecommand{\lemmaname}{Lemma}
  \providecommand{\propositionname}{Proposition}
  \providecommand{\remarkname}{Remark}
\providecommand{\theoremname}{Theorem}
\providecommand{\conjecturename}{Conjecture}
\providecommand{\assumptionname}{Assumption}
\providecommand{\algorithmname}{Algorithm}

\begin{document}

 \title{Minimization of Age-of-Information in Remote Sensing with Energy Harvesting
 \thanks{Akanksha Jaiswal is with the Department of Electrical Engineering, Indian Institute of Technology, Delhi. Email:   akanksha.jaiswal@ee.iitd.ac.in, Arpan Chattopadhyay is with the Department of Electrical Engineering and the Bharti School of Telecom Technology and Management, Indian Institute of Technology, Delhi. Email:   arpanc@ee.iitd.ac.in.  }
 \thanks{This work was supported by the faculty seed grant and professional development allowance (PDA) of IIT Delhi.}
}

 \author{
Akanksha Jaiswal, Arpan Chattopadhyay \\
\vspace{-0.4in}
 }

\maketitle

\begin{abstract} 
 In this paper, minimization of time-averaged age-of-information (AoI) in an energy harvesting (EH) source  equipped remote sensing setting is considered.  The EH source opportunistically samples one or multiple processes over discrete time instants, and sends the status updates to a sink node over a time-varying wireless link. At any discrete time instant, the EH node decides whether to probe the link quality using its stored energy, and further decides whether to sample a process and communicate the data based on the channel probe outcome.   The trade-off is between the freshness of information available at the sink node and the available energy at the energy buffer of the source node. To this end, an infinite horizon Markov decision process theory is used to formulate the problem of minimization of time-averaged expected  AoI for a single energy harvesting source node. The following two   scenarios are considered:  (i) single process with channel state information at transmitter (CSIT),  (ii) multiple processes with  CSIT.  In each scenario, for probed channel state, the optimal source node sampling policy is shown to be a threshold policy involving the instantaneous age of the process(es), the available energy in the buffer and the instantaneous channel quality as the decision variables. Finally, numerical results are provided to demonstrate the policy structures and trade-offs.
\end{abstract}


\section{Introduction}\label{section:introduction}
In recent years, the need for combining the physical systems with the cyber-world has attracted significant research interest. These cyber-physical systems (CPS) are supported by  ultra-low power, low latency IoT networks, and encompass a  large number of applications such as vehicle tracking, environment monitoring, intelligent transportation, industrial process monitoring, smart home systems etc. Such systems often require deployment of sensor nodes to monitor a physical process and send the real time status updates to a  remote estimator over a wireless network. However, for such critical CPS applications, minimizing mean packet  delay without accounting for delay jitter can often be detrimental for the system performance. Also, mean delay minimization does not guarantee  delivery of the observation packets to the sink node in the same order in which they were generated, thereby often  resulting in unnecessarily dedicating network resources towards delivering outdated observation packets despite the availability of a freshly generated observation packet in the system. Hence, it is necessary to take into account the freshness of information of the data packets, apart from the mean packet delay.  

Recently, a metric named  Age of Information (AoI) has been  proposed \cite{kaul2012real} as a  measure of the freshness of the information update. In this setting, a sensor  monitoring a system generates time stamped status update and sends it to the sink over a network. At time $t$, if the latest monitoring information available to the sink node comes from a packet whose time-stamped generation instant was $t'$, then the AoI at the sink node is computed as $(t-t’)$. Thus, AoI has emerged as an alternative performance metric to mean delay \cite{talak2018can}.

 \begin{figure}[t!]
    \begin{center}
  \includegraphics[height=2.5cm,width=7cm]{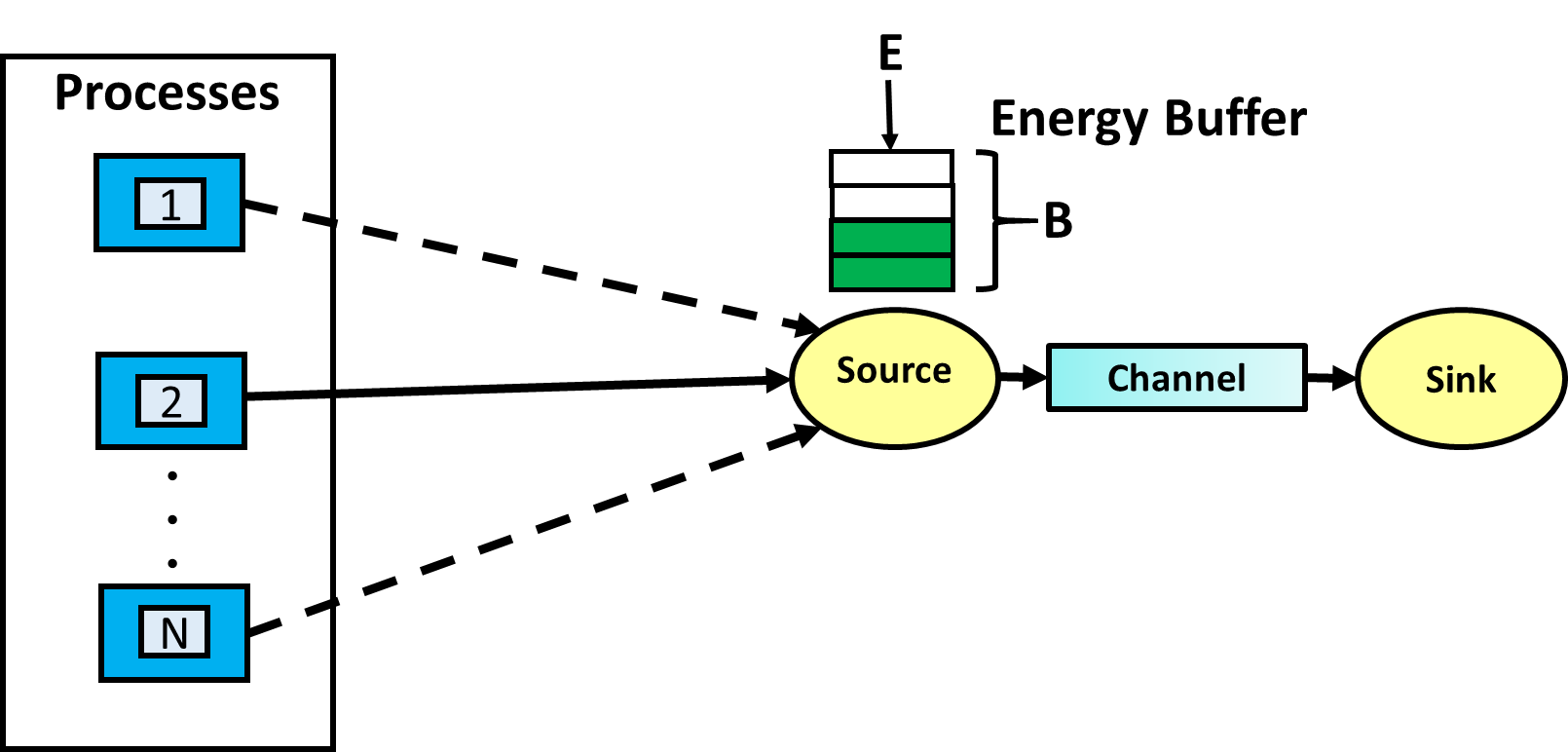}
  \caption{Pictorial representation of a remote sensing system where an EH source samples one of $N$ number of processes at a time and sends the observation packet to a sink node.}
  \label{system-model}
  \end{center}
  \vspace{-8mm}
\end{figure}

 However, timely delivery of the status updates is often limited by energy and  bandwidth constraints in the network. Recent efforts towards designing EH source nodes ({\em e.g.}, source nodes equipped with solar panels) have opened a new research paradigm for IoT network operations.  The energy generation process in such nodes are very uncertain and typically modeled as a stochastic process. The harvested energy is stored in an energy buffer as energy packets, and used for sensing and communication as and when needed. This EH capability significantly improves network lifetime and eliminates the need for frequent manual battery replacement, but poses a new challenge towards network operations due to uncertainty in the  available energy at the source nodes at any given time. 
 
 Motivated by the above challenges, we consider the problem of minimizing the time-averaged expected AoI in a remote sensing setting, where a single EH source  probes the channel state, samples one or multiple processes and sends the observation packets to the sink node over a fading channel. Energy generation process is modeled as a discrete-time i.i.d. process, and  a finite energy buffer is considered. Two variants of the problem are considered:  (i) single process with CSIT,  (ii) multiple processes with CSIT. Channel state probing and process sampling for time-averaged expected AoI minimization problem is formulated as an MDP, and the threshold nature of the optimal policy is established analytically for each case.  Numerical results validate the theoretical results  and intuitions.

 \vspace{-2.5mm}
 \subsection{Related work} 
  Initial efforts towards optimizing AoI mostly involved the analysis of various  queueing  models; {\em e.g.,}  \cite{kaul2012real} for analysing a  single source single server queueing system with FCFS service discipline, \cite{kaul2012status} for LCFS service discipline for M/M/1 queue,  \cite{yates2018age} and \cite{kaul2020timely} for multi-source single sink system with M/M/1 queueing at each source, \cite{kosta2019age}  for  AoI performance analysis for multi-source single-sink infinite-buffer queueing system where unserved packets  are substituted by available newer ones, etc. 
  
  On the other hand, a number of papers   have considered   AoI minimization problem under EH setting:  \cite{farazi2018age} for derivation of average AoI for a single source having finite battery capacity,  \cite{arafa2017age} for derivation of  the minimal age policy for EH two hop network, \cite{farazi2018average} for  average AoI expression for single source EH server,  \cite{hu2018age} for AoI minimization for wirelessly powered user, \cite{arafa2017age1} for  sampling, transmission scheduling and transmit power selection for single source single sink system over infinite time horizon  where delay is dependent on packet transmit energy.

 There have also been several other works on developing optimal scheduling policy for  minimizing AoI for EH sensor networks  \cite{wu2017optimal,yang2015optimal,feng2018age,bacinoglu2015age,leng2019age,gindullina2020age,ceran2019reinforcement,arafa2018age,bacinoglu2017scheduling}.  For e.g., \cite{wu2017optimal} has investigated optimal online policy for single sensor single sink system with a noiseless channel, for infinite, finite and unit  size battery capacity; for the finite battery size, it has provided energy aware status update policy. The paper 
 \cite{yang2015optimal} has considered a multi-sensor single sink system with {\em infinite} battery size, and proposed a randomized myopic scheduling scheme. 
 In \cite{feng2018age}, the optimal online status update policy to minimize the long run average AoI for EH source with updating erasures have been proposed. It has been shown that the best effort uniform updating policy is optimal when there is no feedback and best-effort uniform updating with retransmission (BUR) policy is optimal when feedback is available to the source. The authors of 
 \cite{bacinoglu2015age}  examined the problem of minimizing AoI under a constraint on the count of status updates.  The authors of  \cite{leng2019age} addressed AoI minimization problem for cognitive radio communication system with EH capability; they formulated  optimal sensing and updating for perfect and imperfect sensing as a partially observable Markov decision process (POMDP). Information source diversity, {\em i.e.,}  multiple sources tracking the same process but with dissimilar energy cost, and sending status updates to the EH monitoring node with finite battery size, has been considered in \cite{gindullina2020age} with an MDP formulation, but no structure was provided for the optimal policy. In \cite{ceran2019reinforcement}, reinforcement learning has been used to minimize AoI for a single EH sensor with HARQ protocol, but no clear intuition on the policy structure was provided.  The authors of \cite{bacinoglu2017scheduling} have developed a threshold policy for minimizing AoI for a single sensor single sink system with erasure channel and no channel feedback. For a system with Poisson energy arrival, unit battery size and error-free channel, it has shown that a threshold policy achieves average age lower than that of zero-wait policy; based on this,  lower bound on average age for general battery size and erasure channel has been derived. In \cite{abd2020reinforcement}, the authors have proposed optimal sampling threshold policy using MDP formulation for a system consists of multiple sources RF powered by the destination.   

 \vspace{-2.5mm}

\subsection{Our contributions and organization}
\begin{enumerate}
    \item We formulate the problem of minimizing the time-averaged expected AoI in an EH remote sensing system  with a single source monitoring one or multiple processes, as an MDP with two stage action model, which is different from standard MDP in the literature.  Under the assumptions of i.i.d. time-varying channel with CSIT, channel state probing capability at the source, and finite battery size, we derive the optimal policy structures which turn out to be simple threshold policies. The source node, depending on the current age of a process, decides whether to probe  the channel or not. Afterwards, based on the channel probe  outcome, the source node decides  whether to sample the process and send an observation packet, or to remain idle. Thus, the MDP involves taking action in two stages at each time instant.
     \item  We prove convergence of an analogue of  value iteration for this two-stage MDP.
    \item Numerical analysis shows that the threshold for multiple processes turns out to be a function of the relative age of the processes.
    \item We also prove certain interesting properties of various cost functions and some properties of the thresholds as a function of energy and age of the process.\\
    
\end{enumerate}
The rest of  the paper is organized as follows. System model has been explained in Section~\ref{section:system-model}. AoI minimization for the single source single process case is addressed in Section~\ref{section:single-sensor-single-process}.  AoI minimization policy for multiple process sensing is provided in Section~\ref{section:single-sensor-multiple-process}. Numerical results are provided in Section~\ref{section:numerical-work}, followed by the conclusions in Section~\ref{section:conclusion}. All proofs are provided in the appendices.

\section{System model}\label{section:system-model} 
We consider an EH source capable of sensing one out of $N$ different processes at a time,  and reporting the observation packet to a  sink node over a fading channel; see  Figure~\ref{system-model}. Time is discretized with the discrete time index $t \in \{0, 1,2,3, \cdots\}$. At each time, the  source node can decide whether to estimate the quality of the channel  from the source to the sink, or not. If the source node decides to probe the channel state, it can further decide whether to sample a process and communicate the data packet to the sink, or not, depending on the instantaneous channel quality. The source has a finite energy buffer of size $B$ units, where $E_{p}$ unit of buffer energy is used to probe channel state information and $E_{s}$ unit of buffer energy is used in sensing and communication. The energy packet generation process in the energy buffer is assumed to be an i.i.d. process with known mean. In case the energy buffer in a source node is full, the newly generated energy packets will not be accommodated unless $E_{p}$ unit of energy packet is spent in probing. Let  $A(t)$ denote the   number of energy packet arrivals to the energy buffer at time~$t$, and $E(t)$ denote the energy available to the source at time~$t$, for all $t \geq 0$.

We denote by $p(t)$ the probability of packet transmission success from the source to the sink node  at time~$t$. In this paper, we consider fading channel where $p(t) \in \{p_1, p_2, \cdots, p_m\}$ is i.i.d. across $t$, with $\mathbb{P}(p(t)=p_j)=q_j$ for all $j \in \{1,2,\cdots,m\}$. The channel state corresponding to channel success probability $p_j$ is denoted by $C_j$, and  the packet success probability corresponding to channel state $C_j$ is given by $p(C_j)=p_j$. Let us also denote by $r(t) \in \{0,1\}$ the indicator that the packet transmission from the source to the sink at time~$t$ is successful. Hence,   $\mathbb{P}(r(t)=1|C(t)=C_j)=p_j$ for all $j \in \{1,2,\cdots,m\}$. It is assumed that the channel state $C(t)$ is learnt perfectly via a channel probe.

At time~$t$, let $b(t)\in \{0,1\}$ denote the indicator of deciding to probe the channel, and $a(t) \in \{0,1,\cdots,N\}$ denote the identity of the process being sampled, with $a(t)=0$ meaning that no process is sampled, and $b(t)=0$ meaning that the channel is not probed. Also, $b(t)=0$ implies $a(t)=0$. The set of possible actions or decisions is denoted by $\mathcal{A}=\{\{0,0\}\cup \{1\times \{0,1,2,\cdots,N\}\}\}$, where a generic action at time~$t$ is denoted by $(b(t),a(t))$.

Let us denote by $\tau_k(t)\doteq \sup\{0 \leq \tau < t: a(\tau)=k, r(\tau)=1\}$ the last time instant before time~$t$, when process~$k$ was  sampled and the observation packet was successfully delivered to the sink. The age of information (AoI) for the $k$-th process at time~$t$ is   given by $T_k(t)=(t-\tau_k(t))$. However, if $a(t)=k$ and $r(t)=1$, then $T_k(t)=0$ since the current observation of the $k$-th process is available to the sink node.

A generic scheduling policy is a collection of mappings $\{\mu_t\}_{t \geq 0}$ from the available energy level, probed channel capability, and process sampling and data transmission  history summarized in the AoI of various processes,  to $\mathcal{A}$, which basically decides the decision rule at each time. Thus, the decision rule $\mu_t$ at time~$t$ takes the current state $s(t) $ as input and maps it to one decision in the action space $\mathcal{A}$. If $\mu_t=\mu$ for all $t \geq 0$, the policy is called stationary, else non-stationary. 

We seek to find a stationary scheduling policy $\mu$ that minimizes the expected AoI, summed over nodes and averaged over time. In other words, we seek to solve the following mathematical problem:
\begin{align}\label{eqn:main-problem}
\footnotesize
 \min_{\mu} \frac{1}{T } \sum_{t=0}^T \sum_{k=1}^N \mathbb{E}_{\mu} (T_k(t))
 \normalsize
\end{align}

\section{Single source sensing single process} \label{section:single-sensor-single-process}
In this section, we derive the optimal channel probing, source activation and data transmission policy for a single EH source sampling a single process ($N=1$), which will provide   insights to develop process sampling policy for $N>1$.

 Here, we formulate (1) as a long-run average cost MDP with state space $\mathcal{S} \doteq \{0,1,\cdots,B\} \times \mathbb{Z}_+$ and an intermediate state space $\mathcal{V}=\{0, 1,…..,B\}\times \mathbb{Z}_{+} \times \{C_{1}, C_{2},\cdots, C_{m}\}$ where a generic state $s = (E,T)$ which means that the energy buffer has $E$ energy packets, and  the source was last activated $T$ slots ago. A generic intermediate state $v = (E,T,C)$ which additionally means that the current channel state $C$, obtained via probing, has packet success probability $p(C)$. The action space $\mathcal{A}=\{\{0,0\}\cup \{1\times\{0,1 \}\}\}$ with  $a(t), b(t) \in \{0,1\}$. At each time, if the source node decides not to probe the channel state then it will not perform sampling, and thus $b(t) = 0, a(t) = 0$, and the expected single-stage AoI cost  is $c(s(t), b(t), a(t)) = T$. However, if the source node decides to probe the channel state, the expected single-stage AoI cost is $c(v(t), b(t)=1, a(t)=0) = T$, and $c(v(t), b(t)=1, a(t)=1) = T(1-p(C))$, where the expectation is taken over packet success probability $p(C)$. We first formulate the average-cost MDP problem as an $\alpha$-discounted cost MDP problem with $\alpha \in(0, 1)$, and derive  the optimal policy, from which the solution of the average cost minimization problem can be obtained by taking $\alpha \rightarrow 1$. 

\subsubsection{Optimality equation}
Let $J^{*}(E,T)$ be the optimal value function for state $(E,T)$ in  the discounted cost problem, and let $W^*(E,T,C)$ be the cost-to-go from an intermediate state $(E,T,C)$. The Bellman equations are given by:
\scriptsize
\begin{eqnarray} \label{eqn:Bellman-eqn-single-sensor-single-process-with-fading-general}
J^{*}(E \geq E_{p}+E_{s},T)&=&min \bigg\{T+\alpha \mathbb{E}_{A}J^{*}(min\{E+A,B\},T+1), \nonumber\\
&&V^{*}(E,T) \bigg\}\nonumber\\
V^{*}(E,T)&=& \sum_{j=1}^m q_{j} W^{*}(E,T,C_{j})\nonumber\\ 
& \nonumber\\
& \nonumber\\
W^{*}(E,T,C)&=&min\{T+\alpha  \mathbb{E}_{A}J^{*}(E-E_{p}+A,T+1),\nonumber\\ 
&&T(1-p(C))+\alpha p(C)\mathbb{E}_{A}J^{*}(E-E_{p}-E_{s}+\nonumber\\
&&A,1)+\alpha (1-p(C))\mathbb{E}_{A}J^{*}(E-E_{p}-E_{s}+A,\nonumber\\
&&T+1)\}\nonumber\\
J^{*}(E< E_{p}+E_{s},T)&=&T+\alpha \mathbb{E}_{A}J^{*}(min\{E+A,B\},T+1)
\end{eqnarray}
\normalsize

The first expression in the minimization in the R.H.S. of the first equation in \eqref{eqn:Bellman-eqn-single-sensor-single-process-with-fading-general} is the cost of  not probing channel state ($b(t)=0$), which includes single-stage AoI cost $T$  and an $\alpha$ discounted future cost with a random next state $(min\{E+A,B\},T+1)$, averaged over the distribution of the number of energy packet generation $A$. The quantity $V^{*}(E,T)$ is the expected cost of probing the channel state, which explains the second equation in \eqref{eqn:Bellman-eqn-single-sensor-single-process-with-fading-general}. At an intermediate state $(E,T,C)$, if $a(t)=0$, a single stage AoI cost $T$ is incurred and the next state becomes $(E-E_p+A,T+1)$; if $a(t)=1$, the expected AoI cost is $T(1-p(C))$ (expectation taken  over the packet success probability $p(C)$),  and the next random state becomes $(E-E_p-E_s+A,1)$ and $(E-E_p-E_s+A,T+1)$ if $r(t)=1$ and $r(t)=0$, respectively. The last equation in \eqref{eqn:Bellman-eqn-single-sensor-single-process-with-fading-general} follows similarly since $b(t)=0, a(t)=0$ is the only possible action when $E<E_p+E_s$.  

Substituting the value of $V^{*}(E,T)$ in the first equation of  \eqref{eqn:Bellman-eqn-single-sensor-single-process-with-fading-general}, we obtain the following Bellman equations:

\scriptsize
\begin{eqnarray}\label{eqn:Bellman-eqn-single-sensor-single-process-with-fading}
 J^{*}(E \geq E_{p}+E_{s},T)&=&min \bigg\{T+\alpha \mathbb{E}_{A}J^{*}(min\{E+A,B\},T+1),\nonumber\\ 
&&\mathbb{E}_{C} \bigg( min\{T+\alpha  \mathbb{E}_{A}J^{*}(E-E_{p}+A,T+1),\nonumber\\ 
&&T(1-p(C))+\alpha p(C)\mathbb{E}_{A}J^{*}(E-E_{p}-E_{s}+\nonumber\\
&&A,1)+\alpha (1-p(C))\mathbb{E}_{A}J^{*}(E-E_{p}-E_{s}+\nonumber\\
&&A,T+1)\} \bigg) \bigg\} \nonumber\\
J^{*}(E < E_{p}+E_{s},T)&=&T+\alpha \mathbb{E}_{A}J^{*}(min\{E+A,B\},T+1) 
\end{eqnarray}

\normalsize

\subsubsection{Policy structure}
\begin{proposition} \label{proposition:convergence-of-value-function-J}
The value function $J^{(k)}(s)$ converges to $J^{*}(s)$  as k tends to $\infty$.
\end{proposition}

\begin{proof}See Appendix ~\ref{appendix:proof-of-convergence-of-value-function-J}.
\end{proof}
We provide the convergence proof of value iteration since we have a two-stage decision process as opposed to traditional MDP where a single action is taken.
\begin{lemma} \label{lemma:single-sensor-single-process-J-decreasing-in-p}
For $N=1$, the value function 
$J^{*}(E, T)$ is increasing in $T$ and $W^{*}(E,T,C)$ is decreasing in $p(C)$.  
\end{lemma}

\begin{proof} See Appendix ~\ref{appendix:proof-of-lemma-single-sensor-single-process-fading-cost-increasing-in-T-decrease-in-p}.
\end{proof}
\begin{conjecture}\label{conjecture:single-sensor-single-process-with-fading-policy-structure}
 For $N=1$, the optimal probing policy for the $\alpha$-discounted AoI cost minimization problem is a threshold policy on $T$. For any $E \geq E_{p}+E_{s} $, the optimal action is to probe the channel state if and only if $T\geq T_{th}(E)$ for a threshold function $T_{th}(E)$ of $E$. 
\end{conjecture}

\begin{theorem}\label{theorem:single-sensor-single-process-with-fading-policy-structure}
 For $N=1$, at any time, if the source decides to probe the channel, then the optimal sampling policy  is a threshold policy on $p(C)$. For any $E \geq E_{p}+E_{s} $ and probed channel state, the optimal action is to sample the source node if and only if $p(C)\geq p_{th}(E,T)$ for a threshold function $p_{th}(E,T)$ of $E$ and $T$. 
\end{theorem}
\begin{proof} See Appendix ~\ref{appendix:proof-of-theorem-single-sensor-single-process-fading-threshold-policy}.
\end{proof}

The policy structure supports the following two intuitions. Firstly, given $E$ and $T$,  the source decides to probe the channel state   if AoI is greater than some threshold value. 
Secondly, given $E$ and $T$ and probed channel state, if the channel quality is better than a threshold, then the optimal action is to sample the process and communicate the observation to the sink node. We will later numerically observe in Section~\ref{section:numerical-work} some intuitive properties of $T_{th}(E)$ as a function of $E$ and $\lambda$, and $p_{th}(E,T)$ as a function of $E$, $T$ and $\lambda$.
 
 \vspace{1mm}
\section{Single source sensing multiple processes}\label{section:single-sensor-multiple-process}


In this section, we find the optimal policy for channel probing, source activation, process sampling and data communication in order to solve Problem~\eqref{eqn:main-problem}, when a single source can sample $N$ different processes, one at a time. 

Here we formulate the $\alpha$-discounted cost version of \eqref{eqn:main-problem} as an MDP with a generic state  $s=(E,T_1,T_2,\cdots, T_N)$ which means that the energy buffer has $E$ energy packets, and  the k-th process was last activated $T_{k}$ slots ago. Also, a generic intermediate state $v=(E,T_1,T_2,\cdots, T_N,C)$ which additionally means that the current channel state $C$ is learnt by probing, has packet success probability $p(C)$. The action space  $\mathcal{A}=\{\{0,0\}\cup \{1\times\{0,1,2,\cdots,N \}\}\}$ with $b(t)\in \{0,1\}$ and $a(t) \in \{0,1,\cdots,N\}$. At each time, if the source node decides not to probe the channel state then it will not sample any process, thus $b(t)=0, a(t)=0$ and the expected single stage AoI cost is $c(s(t),b(t), a(t))=\sum_{i=1}^N T_i$. However, if the source node decides to probe the channel state, the expected single stage  AoI cost is $c(v(t),b(t)=1, a(t)=0)=\sum_{i=1}^N T_i$ and $c(v(t), b(t)=1, a(t)=k)=\sum_{i \neq k} T_i+T_k(1-p(C))$ where the expectation is taken over packet success probability $p(C)$. 
\subsubsection{Optimality equation}
In this case, the Bellman equations are given by:

\scriptsize
\begin{eqnarray}\label{eqn:Bellman-eqn-single-sensor-multi-process-with-fading-general}
&&J^{*}(E \geq E_{p}+E_{s},T_1, T_2,\cdots, T_N)\nonumber\\
&=&min \bigg\{\sum_{i=1}^N T_i+\alpha \mathbb{E}_{A}J^{*}(min\{E+A,B\},T_1+1, T_2+1,\cdots, \nonumber\\
&&T_N+1), V^{*}(E,T_1, T_2,\cdots, T_N)\bigg\}\nonumber\\
&&V^{*}(E,T_1, T_2,\cdots, T_N)\nonumber\\
&=& \sum_{j=1}^m q_{j} W^{*}(E,T_1, T_2,\cdots, T_N,C_{j})\nonumber\\ 
&&W^{*}(E,T_1, T_2,\cdots, T_N,C)\nonumber\\ 
&=& min\{\sum_{i=1}^N T_i+\alpha \mathbb{E}_{A}J^{*}(E-E_{p}+A,T_1+1, T_2+1, \cdots,  T_N+1),\nonumber\\ 
&&\min_{1 \leq k \leq N} \bigg( \sum_{i \neq k } T_i+T_k(1-p(C))+\alpha p(C)\mathbb{E}_{A}J^{*}(E-E_{p}-E_{s}+A,\nonumber\\
&&T_1+1,T_2+1, \cdots, T_k'=1,T_{k+1}+1,\cdots,T_N+1)+\nonumber\\
&&\alpha (1-p(C))\mathbb{E}_{A}J^{*}(E-E_{p}-E_{s}+A,T_1+1,T_2+1, \cdots, T_N+1)\bigg)\}\nonumber\\
&&J^{*}(E < E_{p}+E_{s},T_1, T_2,\cdots, T_N)\nonumber\\
&=&\sum_{i=1}^N T_i+\alpha \mathbb{E}_{A}J^{*}(min\{E+A,B\},T_1+1, T_2+1, \cdots,T_N+1) 
\end{eqnarray}
\normalsize

The first expression in the minimization in the R.H.S. of the first equation in \eqref{eqn:Bellman-eqn-single-sensor-multi-process-with-fading-general} is the cost of  not probing channel state ($b(t)=0$), which includes single-stage AoI cost $\sum_{i=1}^N T_i$ and an $\alpha$ discounted future cost with a random next state $( \min\{E+A,B\},T_1+1, T_2+1, \cdots,T_N+1 )$, averaged over the distribution of the number of energy packet generation $A$. The quantity $V^{*}(E,T_1, T_2,\cdots, T_N)$ is the optimal expected cost of probing the channel state, which explains the second equation in
\eqref{eqn:Bellman-eqn-single-sensor-multi-process-with-fading-general}. At an intermediate state $(E,T_1, T_2,\cdots, T_N,C)$, if $a(t)=0$,
 a single stage AoI cost $\sum_{i=1}^N T_i$  is incurred and the next state becomes $(E-E_{p}+A,T_1+1, T_2+1, \cdots,T_N+1)$; if $a(t)=k$, the expected AoI cost is $\sum_{i \neq k } T_i+T_k(1-p(C))$ (expectation taken  over the packet success probability $p(C)$),  and the next random state becomes  $(E-E_{p}-E_{s}+A,T_1+1,T_2+1, \cdots, T_k'=1,   T_{k+1}+1,\cdots,T_N+1)$ and $(E-E_{p}-E_{s}+A,T_1+1,T_2+1, \cdots, T_N+1)$ if $r(t)=1$ and $r(t)=0$, respectively. The last equation in \eqref{eqn:Bellman-eqn-single-sensor-multi-process-with-fading-general} follows similarly since $b(t)=0, a(t)=0$ is the only possible action when $E<E_p+E_s$.

Substituting the value of $V^{*}(E,T_1, T_2,\cdots, T_N)$ in the first equation of  \eqref{eqn:Bellman-eqn-single-sensor-multi-process-with-fading-general}, we obtain the Bellman equations
\eqref{eqn:Bellman-eqn-single-sensor-multiple-process-with-fading}

\begin{figure*}
\scriptsize
\begin{eqnarray}\label{eqn:Bellman-eqn-single-sensor-multiple-process-with-fading}
J^{*}(E \geq E_{p}+E_{s},T_1, T_2,\cdots, T_N)&=&min \bigg\{\sum_{i=1}^N T_i+\alpha \mathbb{E}_{A}J^{*}(min\{E+A,B\},T_1+1, T_2+1, \cdots,T_N+1),  \nonumber\\
&&\mathbb{E}_{C} \bigg( min\{\sum_{i=1}^N T_i+\alpha \mathbb{E}_{A}J^{*}(E-E_{p}+A,T_1+1, T_2+1, \cdots,  T_N+1),\nonumber\\ 
&&\min_{1 \leq k \leq N} \bigg( \sum_{i \neq k } T_i+T_k(1-p(C))+\alpha p(C)\mathbb{E}_{A}J^{*}(E-E_{p}-E_{s}+A,T_1+1,\nonumber\\
&&T_2+1, \cdots, T_k'=1,T_{k+1}+1,\cdots,T_N+1)\nonumber\\
&&+\alpha (1-p(C))\mathbb{E}_{A}J^{*}(E-E_{p}-E_{s}+A,T_1+1,T_2+1, \cdots, T_N+1) \bigg)\} \bigg) \bigg\} \nonumber\\
J^{*}(E < E_{p}+E_{s},T_1, T_2,\cdots, T_N)&=&\sum_{i=1}^N T_i+\alpha \mathbb{E}_{A}J^{*}(min\{E+A,B\},T_1+1, T_2+1, \cdots,T_N+1) 
\end{eqnarray}
\normalsize
\hrule
\end{figure*}

\subsubsection{Policy structure}

\begin{lemma}\label{lemma:single-sensor-multiple-process-J-increasing-in-T-decreasing-in-p(C^')}
 For $N>1$, the value function $J^*(E, T_1,T_2, \cdots, T_N)$ is increasing in each of $T_1, T_2, \cdots, T_N $ and $W^{*}(E,T_1, T_2,\cdots, T_N,C)$ is decreasing in $p(C)$. 
 \end{lemma}
 
\begin{proof} See Appendix ~\ref{appendix:proof-of-lemma-single-sensor-multiple-process-fading-cost-increasing-in-T-decrease-in-p}.
\end{proof}
Let us define $\bm{T}=[T_1, T_2,\cdots, T_N]$ and $\bm{T}_{-k}=[T_1, T_2,\cdots, T_{k-1}, T_{k+1},\cdots, T_N]$. 
\begin{conjecture}\label{theorem:single-sensor-multiple-process-with-fading-policy-structure}
 For $N>1$, the optimal probing policy for the $\alpha$-discounted AoI cost minimization problem is a threshold policy on $\arg \max_{1 \leq k \leq N}T_{k}$. For any $E \geq E_{p}+E_{s} $, the optimal action is to probe the channel state if and only if $\arg \max_{1 \leq k \leq N}T_{k}\geq T_{th}(E,\bm{T}_{-k})$ for a threshold function $T_{th}(E,\bm{T}_{-k})$ of $E$ and $\bm{T}_{-k}$.

 \end{conjecture}
 
\begin{theorem}\label{theorem:single-sensor-multiple-process-with-fading-policy-structure}
 For $N>1$, after probing the channel state, the optimal source activation policy for the $\alpha$-discounted cost problem is a threshold policy on p(C). For any $E \geq E_{p}+E_{s} $ and probed channel state, the optimal action is to sample the process $\arg \max_{1 \leq k \leq N} T_k$ if and only if $p(C) \geq p_{th}(E, \bm{T})$ for a threshold function $p_{th}(E, \bm{T})$ of $(E, \bm{T})$.
 
\end{theorem}


\begin{proof} See Appendix ~\ref{appendix:proof-of-theorem-single-sensor-multiple-process-fading-threshold-policy}.
\end{proof}

The policy structure upholds the two intuitions, (i) for a given $E$ and $(T_1,T_2,\cdots,T_N)$, the source decides to probe the channel state if  highest AoI is greater than some threshold value, 
(ii) for a given $E$ and $(T_1,T_2,\cdots,T_N)$ and probed channel state, if the channel condition is better than a threshold, then the optimal action is to sample the process with highest AoI and send the observation to the sink node. We will later numerically demonstrate some intuitive properties of $T_{th}(E,\bm{T}_{-k})$ as a function of $E$, $\bm{T}_{-k}$ and $\lambda$ and $p_{th}(E,T_1,T_2,\cdots,T_N)$ as a function of $E$, $(T_1,T_2,\cdots,T_N)$ and $\lambda$ in Section~\ref{section:numerical-work}.

\section{Numerical results}\label{section:numerical-work}

\subsection{ Single source, single process ($N=1$) }\label{subsection:single-sensor-single=process-channel-fading}
We consider five channel states ($m=5$) with channel state occurrence probabilities $\bm{q}=[0.2,0.2,0.2,0.2,0.2]$ and the corresponding packet success probabilities $\bm{p}=[0.9, 0.7, 0.5, 0.3, 0.1]$. Energy arrival process is i.i.d. $Bernoulli(\lambda)$ with energy buffer size $B=12$, $E_{p}=1$ unit, $E_{s}=1$ unit, and the discount factor $\alpha=0.99$. Numerical exploration revealed that, there exists a threshold policy on $T$ in decision making for channel state probing; see Figure ~\ref{fig2}(a). It is observed that this $T_{th}(E)$  decreases with $E$ since higher available energy in the energy buffer allows the EH node to probe the channel state more aggressively. Similar reasoning explains the observation that $T_{th}(E)$ decreases with $\lambda$. For probed channel state, Figure~\ref{fig2}(b)  shows the variation of $p_{th}(E,T)$ with $E,T, \lambda$. It is observed that $p_{th}(E,T)$ decreases with $E$, since the EH node  tries to sample the process more aggressively if more energy is available in the  buffer. Similarly, higher value of $T$ results in aggressive sampling, and hence $p_{th}(E,T)$ decreases with $T$. By similar arguments as before, we can explain the observation that this $p_{th}(E,T)$ decreases with $\lambda$. 

\begin{figure}[htb]
   \begin{center}
  \subfloat[]{\includegraphics[height=2.5cm,width=7cm]{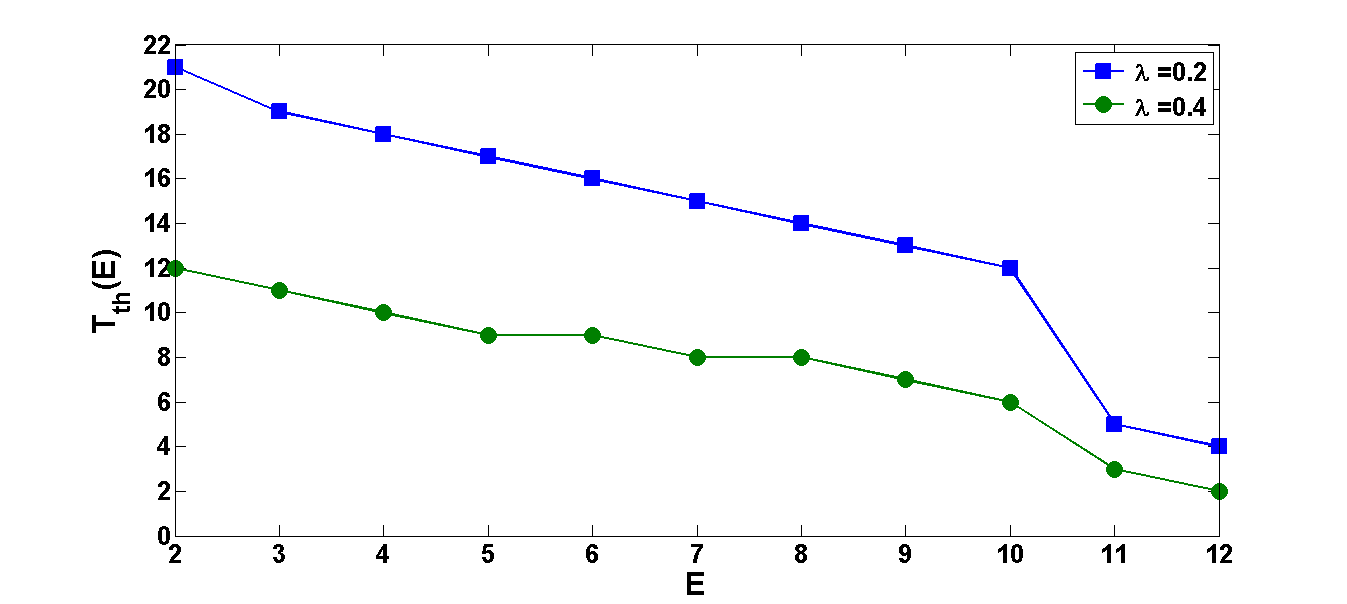}}
     \hfill
   \subfloat[]{\includegraphics[height=2.5cm,width=7cm]{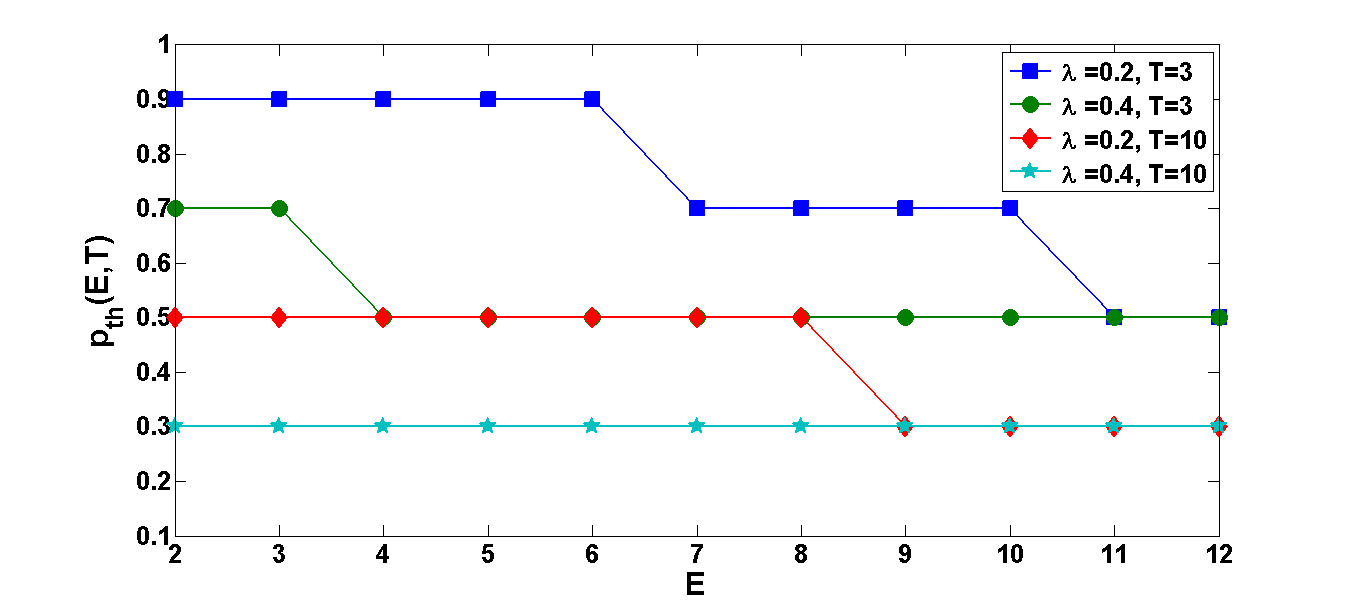}}
  \caption{For $N=1$, (a) Variation of $T_{th}(E)$ with $E, \lambda$  and (b) Variation of $p_{th}(E,T)$ with $E, T, \lambda$.}
  \label{fig2}
 \end{center}
\end{figure}

\vspace{-4mm}

\subsection{ Single source, multiple processes ($N>1$)}
We choose $N=3, \alpha=0.99, B=12, E_{p}=1, E_{s}=1$ and the same channel model and parameters as in Section~\ref{subsection:single-sensor-single=process-channel-fading}. Figure~\ref{fig6}(a) shows the variation on the threshold on $T_1$ for given $T_2, T_3$, for channel state probing. It is observed that  $T_{th}(E, T_2, T_3 )$  decreases with $E$ and $\lambda$. Extensive numerical work also demonstrated that this threshold decreases with each of $T_2$, $T_3$. For probed channel state, Figure~\ref{fig6}(b)  shows  that $p_{th}(E,T_1, T_2, T_3 )$ decreases with $E$ and $\lambda$. Further numerical analysis also demonstrated that this threshold decreases with each of $T_1$, $T_2$, $T_3$ .  

\begin{figure}[htb]
    \begin{center}
  \subfloat[]{\includegraphics[height=2.5cm,width=7cm]{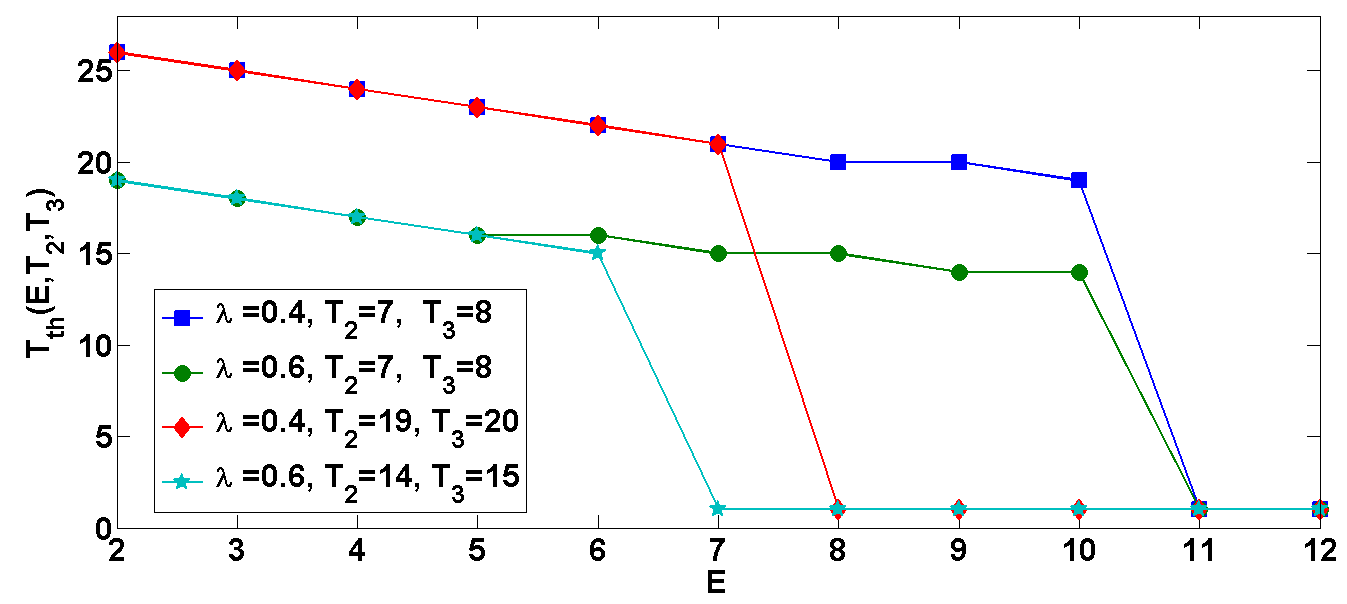}}
  \hfill
  \subfloat[]{  \includegraphics[height=2.5cm,width=7cm]{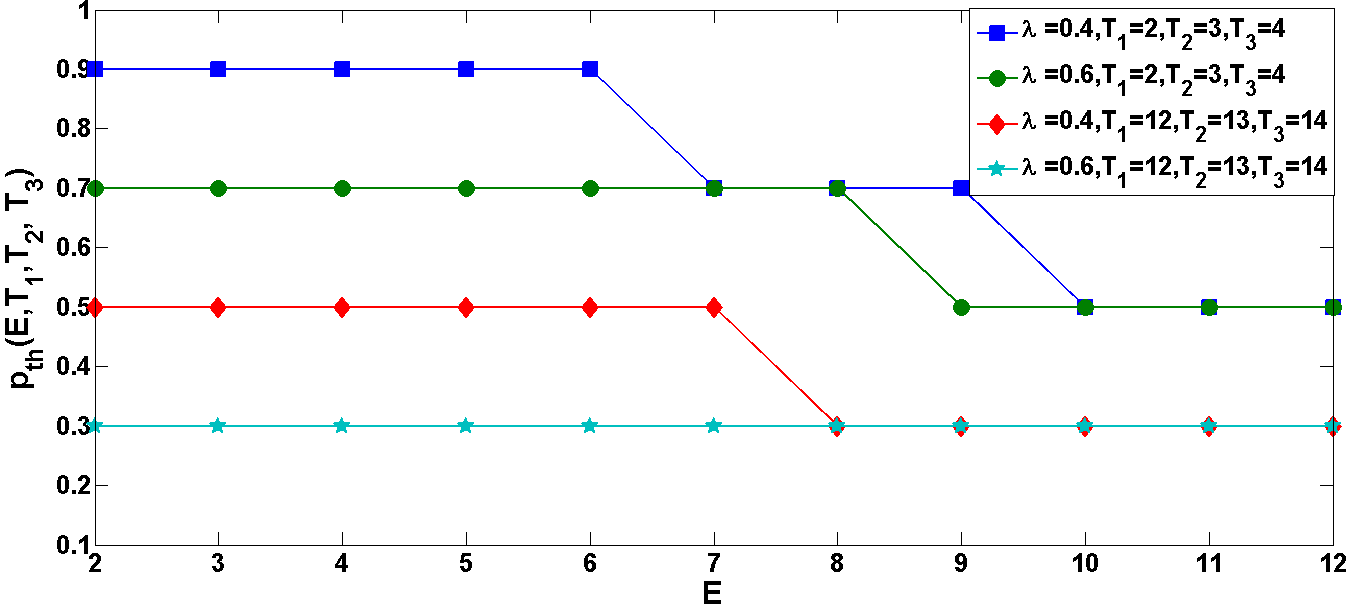}}
  \caption{For $N=3$, (a) Variation of $T_{th}(E,T_2, T_3)$ with $E, T_2,T_3, \lambda$ and (b) Variation of $p_{th}(E,T_1, T_2, T_3)$ with $E,T_1,T_2, T_3, \lambda$.}
  \label{fig6}
  \end{center}
\end{figure}

\vspace{-2mm}

\section{Conclusions}\label{section:conclusion}
In this paper, we  derived the optimal policy structures for minimizing the time-averaged expected AoI under an energy-harvesting source. We    considered  single and multiple processes, i.i.d. time varying channels, and channel probing capability at the source. The optimal source sampling policy turned out to be a threshold policy. Numerical results illustrated the policy structures and trade-offs. We will extend this work for unknown energy generation rate in our future research endeavours.

\newpage
\bibliographystyle{IEEEtran}
\bibliography{ref.bib}
\cleardoublepage
\vspace{12mm}

\appendices
\section {Proof of proposition~\ref{proposition:convergence-of-value-function-J}}\label{appendix:proof-of-convergence-of-value-function-J}
We prove  $max_{s \in \mathcal{S}}|J^{(k)}(s)-J^{*}(s)| \uparrow 0$ as $k \uparrow \infty$. 

Let us define the error $e_{k}=\max_{s \in \mathcal{S}}|J^{(k)}(s)-J^{*}(s)|$\\ and $J^{(0)}(s)$ as initial estimate for $J^{*}(s)$.\\ 
For any state s, we can establish relation between the error at time k+1 to the error at time k in the following way:
\scriptsize
\begin{eqnarray}\label{con1}
&&J^{(k+1)}(E \geq E_{p}+E_{s},T)\nonumber\\
&=& min \bigg\{T+\alpha \mathbb{E}_{A}J^{(k)}(min\{E+A,B\},T+1),  \nonumber\\
&&\mathbb{E}_{C}\bigg(min\{ T+\alpha \mathbb{E}_{A}J^{(k)}(E-E_{p}+A,T+1), \nonumber\\
&&T(1-p(C))+\alpha p(C)\mathbb{E}_{A}J^{(k)}(E-E_{p}-E_{s}+A,1)+ \nonumber\\
&&\alpha (1-p(C))\mathbb{E}_{A}J^{(k)}(E-E_{p}-E_{s}+A,T+1)\} \bigg) \bigg\} 
\end{eqnarray}
\normalsize
We assume there exists an optimal value function $J^{*}(E,T)$ for the discounted cost problem and substituting $J^{(k)}(E,T)$ by $J^{(k)}(E,T) \leq J^{(*)}(E,T) +e_{k}$ in \eqref{con1} we get following equation:

\scriptsize
\begin{eqnarray}\label{con2}
&&J^{(k+1)}(E \geq E_{p}+E_{s},T)\nonumber\\ 
&\leq& min \bigg\{T+ \alpha \mathbb{E}_{A}(J^{*}(min\{E+A,B\},T+1)+e_{k}),\nonumber\\
&&\mathbb{E}_{C}\bigg(min\{ T+\alpha \mathbb{E}_{A}(J^{*}(E-E_{p}+A,T+1)+e_{k}), \nonumber\\
&&T(1-p(C))+\alpha p(C)\mathbb{E}_{A}(J^{*}(E-E_{p}-E_{s}+A,1)+e_{k})\nonumber\\
&&+\alpha (1-p(C))\mathbb{E}_{A}(J^{*}(E-E_{p}-E_{s}+A,T+1)+e_{k})\} \bigg) \bigg\}
\end{eqnarray}

\normalsize
Similarly, substituting $J^{(k)}(E,T)$ by $J^{(k)}(E,T) \geq J^{(*)}(E,T) -e_{k}$ in \eqref{con1} we get following equation:

\scriptsize
\begin{eqnarray}\label{con3}
&&J^{(k+1)}(E \geq E_{p}+E_{s},T)\nonumber\\
&\geq & min \bigg\{T+ \alpha \mathbb{E}_{A}(J^{*}(min\{E+A,B\},T+1)-e_{k}),\nonumber\\
&&\mathbb{E}_{C}\bigg(min\{ T+\alpha \mathbb{E}_{A}(J^{*}(E-E_{p}+A,T+1)-e_{k}), \nonumber\\
&&T(1-p(C))+\alpha p(C)\mathbb{E}_{A}(J^{*}(E-E_{p}-E_{s}+A,1)-e_{k})\nonumber\\
&&+\alpha (1-p(C))\mathbb{E}_{A}(J^{*}(E-E_{p}-E_{s}+A,T+1)-e_{k})\} \bigg) \bigg\}
\end{eqnarray}
\normalsize
Combining the results obtained from equations \eqref{con2} and \eqref{con3}, we get:
\scriptsize
\begin{eqnarray}
&&J^{*}(E \geq E_{p}+E_{s},T)-\alpha e_{k}\nonumber\\
&\leq& J^{(k+1)}(E \geq E_{p}+E_{s},T) \leq J^{*}(E \geq E_{p}+E_{s},T)+\alpha e_{k}
\end{eqnarray}

\begin{eqnarray}
|J^{(k+1)}(E \geq E_{p}+E_{s},T)-J^{*}(E \geq E_{p}+E_{s},T)| \leq \alpha e_{k} 
\end{eqnarray}

\begin{eqnarray}
\max_{s \in \mathcal{S}}|J^{(k+1)}(s)-J^{*}(s)| \leq \alpha e_{k} 
\end{eqnarray}

\begin{eqnarray}\label{con4}
e_{k+1} \leq \alpha e_{k}
\end{eqnarray}

\normalsize
Thus, equation \eqref{con4} gives the relation between the error at time k+1 to the error at time k. 
By backward substitution we get, 
\scriptsize
\begin{eqnarray}\label{final-convergence}
\max_{s \in \mathcal{S}}|J^{(k)}(s)-J^{*}(s)|\leq \alpha^{k}\max_{s \in \mathcal{S}}|J^{(0)}(s)-J^{*}(s)|
\end{eqnarray}
\normalsize
From equation \eqref{final-convergence}, as $k \uparrow \infty$, the error reduces to zero. Hence, $J^{(k)}(s)$ converges to $J^{*}(s)$.

\section{Proof of Lemma~\ref{lemma:single-sensor-single-process-J-decreasing-in-p}}\label{appendix:proof-of-lemma-single-sensor-single-process-fading-cost-increasing-in-T-decrease-in-p}
 We prove this result by value iteration:
 \scriptsize
\begin{eqnarray}
&&J^{(k+1)}(E \geq E_{p}+E_{s},T)\nonumber\\
&=& min \bigg\{T+\alpha \mathbb{E}_{A}J^{(k)}(min\{E+A,B\},T+ 1), \nonumber\\
&&\mathbb{E}_{C}\bigg(min\{ T+\alpha\mathbb{E}_{A}J^{(k)}(E-E_{p}+A, T+1),\nonumber\\
&&T(1-p(C))+\alpha p(C)\mathbb{E}_{A}J^{(k)}(E-E_{p}-E_{s}+ A,1)+\nonumber\\
&&\alpha (1-p(C))\mathbb{E}_{A}J^{(k)}(E-E_{p}-E_{s}+A,T+1)\} \bigg) \bigg\} \nonumber\\
&&J^{(k+1)}(E < E_{p}+E_{s},T)\nonumber\\
&=&T+\alpha \mathbb{E}_{A}J^{(k)}(min\{E+A,B\},T+1)
\end{eqnarray}
\normalsize

Let us start with $J^{(0)}(s) = 0$ for all $s\in S$. Clearly,
$J^{(1)}(E \geq E_{p}+E_{s}, T) =min \{T, \mathbb{E}_{C}(min\{ T,T(1-p(C))\}) \}$ $=min \{T, \mathbb{E}_{C}(T(1-p(C))\} $ and $J^{(1)}(E < E_{p}+E_{s} , T) = T$. Hence, for any given $E$, the value function $J^{(1)}(E, T)$ is an increasing function of $T$ and decreasing function of $p(C)$. As induction hypothesis, we assume that $J^{(k)}(E, T)$ is also increasing function of $T$.

Now, 

\scriptsize
\begin{eqnarray}\label{value-iteration-single-sensor-with-fading-1}
&&J^{(k+1)}(E \geq E_{p}+E_{s},T)\nonumber\\
&=& min \bigg\{T+\alpha \mathbb{E}_{A}J^{(k)}(min\{E+A,B\},T+1),\nonumber\\
&&\mathbb{E}_{C}\bigg(min\{ T+\alpha \mathbb{E}_{A}J^{(k)}(E-E_{p}+A,T+1),\nonumber\\
&& T(1-p(C))+\alpha p(C)\mathbb{E}_{A}J^{(k)}(E-E_{p}-E_{s}+A,1)+ \nonumber\\
&&\alpha (1-p(C))\mathbb{E}_{A}J^{(k)}(E-E_{p}-E_{s}+A,T+1)\} \bigg) \bigg\} 
\end{eqnarray}
\normalsize

We need to show that $J^{(k+1)}(E \geq E_{p}+E_{s},T)$ is also increasing in $T$. The first term inside the minimization operation in   \eqref{value-iteration-single-sensor-with-fading-1} is increasing in $T$, by the induction hypothesis and from the fact that expectation is a linear operation. On the other hand, the second term has linear expectation over channel state and another minimization operator. Also, the first and second terms inside the second minimization operation in \eqref{value-iteration-single-sensor-with-fading-1} are increasing in $T$ by the induction hypothesis and the linearity of expectation operation. Thus, $J^{(k+1)}(E \geq  E_{p}+E_{s},T)$ is also increasing in $T$. By similar arguments, we can claim that $J^{(k+1)}(E < E_{p}+E_{s} ,T)$ is increasing in $T$. Now, since $J^{(k)}(\cdot)\uparrow J^{*}(\cdot)$ as $k \uparrow \infty$ 
by proof of Proposition~\ref {proposition:convergence-of-value-function-J}, $J^{*}(E,T)$ is also increasing in $T$. Hence, the lemma is proved.




\normalsize

\section{Proof of Theorem~\ref{theorem:single-sensor-single-process-with-fading-policy-structure}}\label{appendix:proof-of-theorem-single-sensor-single-process-fading-threshold-policy}

From \eqref{eqn:Bellman-eqn-single-sensor-single-process-with-fading}, it is obvious that for probed channel state the optimal decision for $E \geq  E_{p}+E_{s} $ is to sample the source if and only if the cost of sampling is lower than the cost of not sampling the source,  i.e.,  
$T+\alpha  \mathbb{E}_{A}J^{*}(E-E_{p}+A,T+1) \geq T(1-p(C))+\alpha \mathbb{E}_{A}J^{*}(E-E_{p}-E_{s}+A,T+1)-
\alpha p(C)\bigg(\mathbb{E}_{A}J^{*}(E-E_{p}-E_{s}+A,T+1)-\mathbb{E}_{A}J^{*}(E-E_{p}-E_{s}+A,1)\bigg)$.
Now, by Lemma~\ref{lemma:single-sensor-single-process-J-decreasing-in-p}, $\mathbb{E}_{A}J^{*}(E-E_{p}-E_{s}+A,T+1)-  \mathbb{E}_{A}J^{*}(E-E_{p}-E_{s}+A,1)$ is non-negative. Thus the R.H.S. decreases with $p(C)$, whereas the L.H.S. is independent of $p(C)$. Hence, for probed channel state the optimal action is to sample if and only if $p(C) \geq p_{th}(E,T)$ for some suitable threshold function $p_{th}(E,T)$.

\section{Proof of Lemma~\ref{lemma:single-sensor-multiple-process-J-increasing-in-T-decreasing-in-p(C^')}}\label{appendix:proof-of-lemma-single-sensor-multiple-process-fading-cost-increasing-in-T-decrease-in-p}
The proof is similar to the proof of Lemma~\ref{lemma:single-sensor-single-process-J-decreasing-in-p} and it follows from the convergence of value iteration as given below:
\scriptsize
\begin{eqnarray}\label{value-iteration-single-sensor-multiple-process-with-fading}
&&J^{(k+1)}(E \geq E_{p}+E_{s},T_1, T_2,\cdots, T_N)\nonumber\\
&=&min \bigg\{\sum_{i=1}^N T_i+\alpha \mathbb{E}_{A}J^{(k)}(min\{E+A,B\},T_1+1, T_2+1, \cdots,\nonumber\\
&&T_N+1),\mathbb{E}_{C} \bigg( min\{\sum_{i=1}^N T_i+\alpha \mathbb{E}_{A}J^{(k)}(E-E_{p}+A,T_1+1, T_2+1, \nonumber\\ 
&&\cdots,  T_N+1),\min_{1 \leq k \leq N} \bigg( \sum_{i \neq k } T_i+T_k(1-p(C))+\alpha p(C)\mathbb{E}_{A}J^{(k)}(E\nonumber\\
&&-E_{p}-E_{s}+A,T_1+1,T_2+1, \cdots, T_k'=1,T_{k+1}+1,\cdots,T_N+1) \nonumber\\
&&+\alpha (1-p(C))\mathbb{E}_{A}J^{(k)}(E-E_{p}-E_{s}+A,T_1+1,T_2+1, \cdots,\nonumber\\
&&T_N+1) \bigg)\} \bigg) \bigg\} \nonumber\\
&&J^{(k+1)}(E < E_{p}+E_{s},T_1, T_2,\cdots, T_N)\nonumber\\
&=&\sum_{i=1}^N T_i+\alpha \mathbb{E}_{A}J^{(k)}(min\{E+A,B\},T_1+1, T_2+1, \cdots,T_N+1)\nonumber\\ 
\end{eqnarray}
\normalsize
Let us start with $J^{(0)}(s) = 0$ for all $s\in S$. Clearly,
$J^{(1)}(E \geq  E_{p}+E_{s}, T_1, T_2,\cdots, T_N) = \min \{\sum_{i=1}^N T_i, \mathbb{E}_{C}(min\{\sum_{i=1}^N T_i,\min_ {1 \leq k \leq N} ( \sum_{i \neq k } T_i+T_k(1-p(C)))\}\}$ and $J^{(1)}(E < E_{p}+E_{s}, T_1, T_2,\cdots, T_N) = \sum_{i=1}^N T_i$. Hence, for any given $E$, the value function $J^{(1)}(E,T_1, T_2,\cdots, T_N)$ is an increasing function of  $T_1, T_2, \cdots, T_N $. As induction hypothesis, we assume that $J^{(k)}(E,T_1, T_2,\cdots, T_N)$ is also increasing function of $T_1, T_2, \cdots, T_N $. 

Now,
\scriptsize
\begin{eqnarray}\label{value-iteration-single-sensor-multiple-process-with-fading-1}
&&J^{(k+1)}(E \geq E_{p}+E_{s},T_1, T_2,\cdots, T_N)\nonumber\\
&=&min \bigg\{\sum_{i=1}^N T_i+\alpha \mathbb{E}_{A}J^{(k)}(min\{E+A,B\},T_1+1, T_2+1, \cdots,\nonumber\\
&&T_N+1),\mathbb{E}_{C} \bigg( min\{\sum_{i=1}^N T_i+\alpha \mathbb{E}_{A}J^{(k)}(E-E_{p}+A,T_1+1, T_2+1, \nonumber\\ 
&&\cdots,  T_N+1),\min_{1 \leq k \leq N} \bigg( \sum_{i \neq k } T_i+T_k(1-p(C))+\alpha p(C)\mathbb{E}_{A}J^{(k)}(E\nonumber\\
&&-E_{p}-E_{s}+A,T_1+1,T_2+1, \cdots, T_k'=1,T_{k+1}+1,\cdots,T_N+1) \nonumber\\
&&+\alpha (1-p(C))\mathbb{E}_{A}J^{(k)}(E-E_{p}-E_{s}+A,T_1+1,T_2+1, \cdots,\nonumber\\
&&T_N+1) \bigg)\} \bigg) \bigg\} 
\end{eqnarray}
\normalsize

We seek to show that $J^{(k+1)}(E \geq E_{p}+E_{s} ,T_1, T_2,\cdots, T_N)$ is also increasing in each of $T_1, T_2, \cdots, T_N$. The first term inner to the minimization operation in \eqref{value-iteration-single-sensor-multiple-process-with-fading-1} is increasing in each of $T_1, T_2, \cdots, T_N$, utilizing the induction hypothesis and linear property of expectation operation. On the other hand, the second term has expectation over channel state and another minimization operator. Also, the fist and second term of second minimization operator is increasing in each of $T_1, T_2, \cdots, T_N$ by using induction hypothesis and the linearity of expectation operation. Thus, $J^{(k+1)}(E\geq E_{p}+E_{s} , T_1, T_2,\cdots, T_N)$ is also increasing in each of $T_1, T_2, \cdots, T_N$. 
Similarly, we can assert that $J^{(k+1)}(E < E_{p}+E_{s},T_1, T_2,\cdots, T_N)$ is increasing in each of $T_1, T_2, \cdots, T_N$.
Now, since $J^{(k)}(\cdot)\uparrow J^{*}(\cdot)$ as $k \uparrow \infty$, $J^*(E,T_1, T_2,\cdots, T_N)$ is also increasing in each of $T_1, T_2,\cdots, T_N$. Hence, the lemma is proved.

\section{Proof of Theorem ~\ref{theorem:single-sensor-multiple-process-with-fading-policy-structure}}\label{appendix:proof-of-theorem-single-sensor-multiple-process-fading-threshold-policy}
It is obvious that $J^*(\cdot)$ is invariant to any permutation of $(T_1,T_2,\cdots, T_N)$. Hence, by Lemma~\ref{lemma:single-sensor-multiple-process-J-increasing-in-T-decreasing-in-p(C^')},  $ \arg \min_{1 \leq k \leq N}\bigg(\sum_{i \neq k } T_i+T_k(1-p(C))+\alpha p(C)\mathbb{E}_{A}J^{(k)}(E-E_{p}-E_{s}+A,T_1+1,T_2+1, \cdots, T_k'=1,T_{k+1}+1,\cdots,T_N+1)+\alpha (1-p(C))\mathbb{E}_{A}J^{(k)}(E-E_{p}-E_{s}+A,T_1+1,T_2+1, \cdots,T_N+1) \bigg)=\arg \max_{1 \leq k \leq N} T_k$, i.e., the best process to activate is $k^* \doteq \arg \max_{1 \leq k \leq N} T_k$ in case one process has to be activated. For probed channel state, it is optimal to sample a process if and only if the cost of sampling this process is less than or equal to the cost of not sampling this process, which translates into  $T_{k^*}+\alpha \mathbb{E}_{A}J^{*}(E-E_{p}+A,T_1+1, T_2+1,\cdots,  T_N+1) \geq  T_{k^*}(1-p(C))+
\alpha \mathbb{E}_{A}J^{*}(E-E_{p}-E_{s}+A,T_1+1,T_2+1, \cdots,T_N+1)-\alpha p(C)\bigg(\mathbb{E}_{A}J^{*}(E-E_{p}-E_{s}+A,T_1+1,T_2+1, \cdots,T_N+1)-\mathbb{E}_{A}J^{*}(E-E_{p}-E_{s}+A,T_1+1,T_2+1, \cdots, T_k'=1,T_{k+1}+1,\cdots,T_N+1)\bigg)$. Now, by Lemma~\ref{lemma:single-sensor-multiple-process-J-increasing-in-T-decreasing-in-p(C^')}, $\mathbb{E}_{A}J^{*}(E-E_{p}-E_{s}+A,T_1+1,T_2+1, \cdots,T_N+1)-\mathbb{E}_{A}J^{*}(E-E_{p}-E_{s}+A,T_1+1,T_2+1, \cdots, T_k'=1,T_{k+1}+1,\cdots,T_N+1) $ is non negative. Thus, the  R.H.S. is decreasing in $p(C)$  and the L.H.S. is independent of $p(C)$. Hence, the threshold structure of the optimal sampling policy is proved.
\normalsize

\end{document}